\documentclass[11pt]{article}
\usepackage[margin=1in]{geometry}
\usepackage{graphicx}
\usepackage[utf8]{inputenc}
\usepackage{amsmath}
\usepackage[ruled,vlined,linesnumbered,noresetcount]{algorithm2e}
\usepackage{amsfonts}
\usepackage{amsmath}
\usepackage{amssymb}

\usepackage{amsthm}
\usepackage{mathtools}
\usepackage{capt-of}
\theoremstyle{plain}
\usepackage{bbm}
\usepackage{bm}
\usepackage{xcolor}
\usepackage[shortlabels]{enumitem} 
\usepackage{comment}
\usepackage{tikz}
\usepackage{xcolor}
\usepackage{cite}
\usetikzlibrary{decorations.pathreplacing}
\usepackage[title]{appendix}
\usepackage[framemethod=tikz]{mdframed}

\newtheorem{theorem}{Theorem}
\newtheorem{lemma}{Lemma}[section]

\newtheorem{definition}[lemma]{Definition}

\newtheorem*{claim*}{Claim}
\newtheorem*{proposition*}{Proposition}
\newtheorem*{lemma*}{Lemma}
\newtheorem*{problem*}{Problem}

\newtheorem{innercustomthm}{Lemma}
\newenvironment{custom_lemma}[1]
  {\renewcommand\theinnercustomthm{#1}\innercustomthm}
  {\endinnercustomthm}

\newtheorem{innercustomdef}{Definition}
\newenvironment{custom_def}[1]
  {\renewcommand\theinnercustomdef{#1}\innercustomdef}
  {\endinnercustomdef}

\allowdisplaybreaks[1]

\newcommand{\C}{\mathcal{C}}

\begin{document}

\title{Fair Colorful $k$-Center Clustering\thanks{A preliminary version of this work was presented at the 21st Conference on Integer Programming and Combinatorial Optimization (IPCO 2020). An independent work of Anegg et al. \cite{aneggfairkcenter}, presented at the same venue, gave a 4-approximation for Colorful k-Center with constantly many colors using different techniques. This work is supported by the Swiss National Science Foundation project 200021-184656 ``Randomness in Problem Instances and Randomized Algorithms.''}
}


\author{Xinrui Jia$^{\dagger}$        \and
        Kshiteej Sheth$^{\dagger}$ \and 
        Ola Svensson\footnote{\'{E}cole Polytechnique F\'{e}d\'{e}rale de Lausanne. Emails: \{xinrui.jia, kshiteej.sheth, ola.svensson\}@epfl.ch}
}


\date{}

\maketitle

\begin{abstract}
An instance of \textit{colorful $k$-center} consists of points in a metric space that are colored red or blue, along with an integer $k$ and a coverage requirement for each color.
The goal is to find the smallest radius $\rho$ such that there exist balls of radius $\rho$ around $k$ of the points that meet the coverage requirements.

The motivation behind this problem is twofold. First, from fairness considerations: each color/group should receive a similar service guarantee, and second, from the algorithmic challenges it poses: this problem combines the difficulties of clustering along with the subset-sum problem.
In particular, we show that this combination results in strong integrality gap lower bounds for several natural linear programming relaxations.

Our main result is an efficient approximation algorithm that overcomes these difficulties to  achieve an approximation guarantee of $3$, nearly matching the tight approximation guarantee of $2$ for the classical $k$-center problem which this problem generalizes. 

\end{abstract}

\section{Introduction}

In the \textit{colorful k-center} problem introduced in \cite{DBLP:conf/esa/Bandyapadhyay0P19}, we are given a set of $n$ points $P$ in a metric space partitioned into a set $R$ of red points and a set $B$ of blue points, along with parameters $k$, $r$, and $b$.
The goal is to find a set of $k$ centers $C \subseteq P$ that minimizes $\rho$ so that balls of radius $\rho$ around each point in $C$ cover at least $r$ red points and at least $b$ blue points.
More generally, the points can be partitioned into $\omega$ color classes $\C_1, \dots, \C_{\omega}$, with coverage requirements $p_1, \dots, p_{\omega}$. To keep the exposition of our ideas as clean as possible, we concentrate the bulk of our discussion to the version with two colors. In Section \ref{mult-colors} we show how our algorithm can be generalized for $\omega$ color classes with an exponential dependence on $\omega$ in the running time in a rather straightforward way, thus getting a polynomial time algorithm for constant $\omega$.

This generalization of the classic $k$-center problem has applications in situations where fairness is a concern.
For example, if a telecommunications company is required to provide service to at least 90\% of the people in a country, it would be cost effective to only provide service in densely populated areas.
This is at odds with the ideal that at least some people in every community should receive service.
In the absence of color classes, an approximation algorithm could be ``unfair" to some groups by completely considering them as outliers. 
The inception of fairness in clustering can be found in the recent paper \cite{DBLP:conf/nips/Chierichetti0LV17} (see also~\cite{DBLP:conf/icml/BackursIOSVW19,DBLP:journals/corr/abs-1905-13651}), which uses a related but incomparable notion of fairness.
Their notion of fairness requires \emph{each individual cluster} to have a balanced number of points from each color class, which leads to very different algorithmic considerations and is motivated by other applications, such as ``feature engineering''.

The other motive for studying the colorful $k$-center problem derives from the algorithmic challenges it poses. 
One can observe that it generalizes the \textit{$k$-center problem with outliers}, which is equivalent to only having red points and needing to cover at least $r$ of them.  This outlier version is already more challenging than the classic $k$-center problem:  only recent results give tight $2$-approximation algorithms \cite{DBLP:conf/icalp/ChakrabartyGK16, DBLP:journals/talg/HarrisPST19}, improving upon the $3$-approximation guarantee of \cite{charikar2001algorithms}. In contrast, such algorithms for the classic $k$-center problem have been known since the '80s\cite{hochbaum1985best,gonzalez1985clustering}. That the approximation guarantee of $2$ is tight, even for classic $k$-center, was proved in \cite{hsu1979easy}.

At the same time, a subset-sum problem with polynomial-sized numbers is embedded within the colorful $k$-center problem. 
To see this, consider $n$ numbers $a_1, \ldots, a_n$ and let $A = \sum_{i=1}^n a_i$.  Construct an instance of the colorful $k$-center problem with $r = k\cdot A + A/2$, $b = k\cdot A  - A/2$, and for every $i\in \{1, \ldots, n\}$, a ball of radius one containing $A+a_i$ red points and $A- a_i$ blue points.  These balls are assumed to be far apart so that any single ball that covers two of these balls must have a very large radius.
It is easy to see that the constructed colorful $k$-center instance has a solution of radius one if and only if there is a size $k$ subset of the $n$ numbers whose sum equals $A/2$.

We use this connection to subset-sum to show that the standard linear programming (LP) relaxation of the colorful $k$-center problem has an unbounded integrality gap even after a linear number of rounds of the powerful Lasserre/Sum-of-Squares hierarchy (see Section \ref{SoS}). We remark that  the standard linear programming relaxation gives a $2$-approximation algorithm for the outliers version even without applying lift-and-project methods. Another natural approach for strengthening the standard linear programming relaxation is to add flow-based inequalities specially designed to solve subset-sum problems. However, in Section \ref{flow-section},  we prove that they do not improve the integrality gap due to the clustering feature of the problem. This shows that clustering and the subset-sum problem are intricately related in colorful $k$-center. This interplay makes the problem more complex and prior to our work only a randomized constant-factor approximation algorithm was known when the points are in $\mathbb{R}^2$ with an approximation guarantee greater than $6$~\cite{DBLP:conf/esa/Bandyapadhyay0P19}.

Our main result overcomes these difficulties and we give a nearly tight approximation guarantee:
\begin{theorem}
    There is a $3$-approximation algorithm for the colorful $k$-center problem.
    \label{thm:main}
\end{theorem}
As aforementioned, our techniques can be easily extended to a constant number of color classes but we
restrict the discussion here to two colors.

On a very high level, our algorithm manages to decouple the clustering and the subset-sum aspects. First, our algorithm guesses certain centers of the optimal solution that it then uses to partition the point set into a ``dense'' part $P_d$ and a ``sparse'' part $P_s$. 
The dense part is clustered using a subset-sum instance while the sparse set is clustered using the techniques of Bandyapadhyay, Inamdar, Pai, and Varadarajan~\cite{DBLP:conf/esa/Bandyapadhyay0P19} (see Section \ref{pseudo-approx}). 
Specifically, we use the pseudo-approximation of~\cite{DBLP:conf/esa/Bandyapadhyay0P19} that satisfies the coverage requirements using $k+1$ balls of at most twice the optimal radius.

While our approximation guarantee is nearly tight, it remains an interesting open problem to give a $2$-approximation algorithm or to show that the ratio $3$ is tight. One possible direction is to understand the strength of the relaxation obtained by combining the Lasserre/Sum-of-Squares hierarchy with the flow constraints. While we show that individually they do not improve the integrality gap, we believe that their combination can lead to a strong relaxation. 

\textbf{Independent work.} Independently and concurrently to our work, authors in ~\cite{aneggfairkcenter} obtained a $4$-approximation algorithm for the colorful k-center problem with $\omega = O(1)$ using different techniques than the ones described in this work. Furthermore they show that, assuming $P\neq NP$, if $\omega$ is allowed to be unbounded then the colorful k-center problem admits no algorithm guaranteeing a finite approximation. They also show that assuming the Exponential Time Hypothesis, colorful k-center is inapproximable if $\omega$ grows faster than $\log n$.

\textbf{Organization.} We begin by giving some notation and definitions and describing the pseudo-approximation algorithm in \cite{DBLP:conf/esa/Bandyapadhyay0P19}.
In fact, we then describe a 2-approximation algorithm on a certain class of instances that are \textit{well-separated}, and the 3-approximation follows almost immediately. 
This 2-approximation proceeds in two phases: the first is dedicated to the guessing of certain centers, while the second processes the dense and sparse sets. 

Section \ref{mult-colors} explains the generalization to $\omega$ color classes.
In Section 3 we present our integrality gaps under the Sum-of-Squares hierarchy and additional constraints deriving from a flow network to solve subset-sums.

\section{A 3-Approximation Algorithm}
In this section we present our 3-approximation algorithm. 
We briefly describe the pseudo-approxima- tion algorithm of Bandhyapadhyay et al. \cite{DBLP:conf/esa/Bandyapadhyay0P19} since we use it as a subroutine in our algorithm. 

\textbf{Notation:} We assume that our problem instance is normalized to have an optimal radius of one and we refer to the set of centers in an optimal solution as $OPT$.
The set of all points at distance at most $\rho$ from a point $j$ is denoted by $\mathcal{B}(j, \rho)$ and we refer to this set as a \textit{ball of radius $\rho$ at $j$}. 
We write $\mathcal{B}(j)$ for $\mathcal{B}(j,1)$.
By a \textit{ball of $OPT$} we mean $\mathcal{B}(j)$ for some $j \in OPT$.

\subsection{The Pseudo-Approximation Algorithm}\label{pseudo-approx}

\begin{figure}[t]
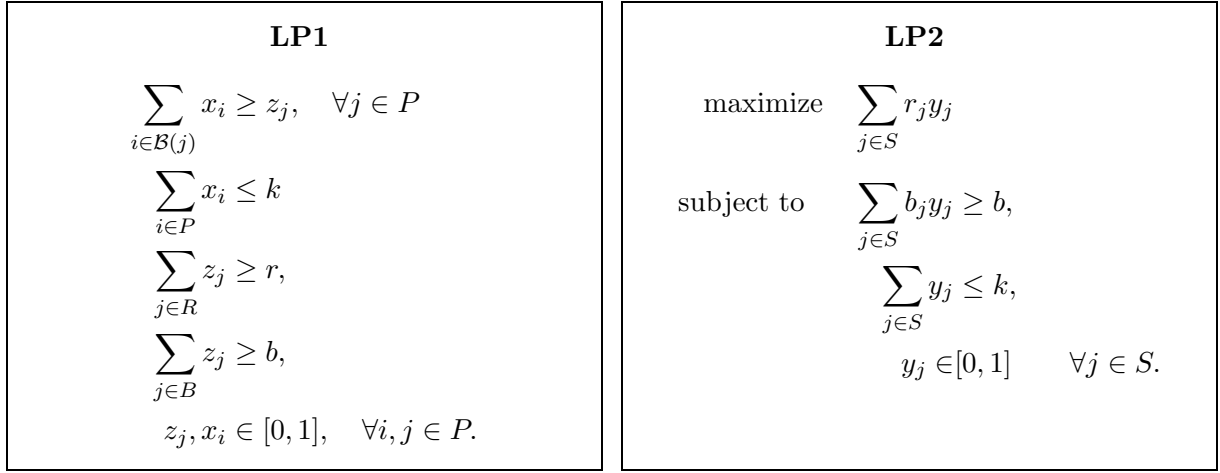

    \fbox{
        \begin{minipage}[t][6.0cm][t]{0.45\textwidth}
        \begin{center}
        \vspace{2mm}
        \textbf{LP1} 
        \vspace{-4mm}\end{center}
\begin{align*}
     \sum_{i \in \mathcal{B}(j)} x_i &\geq z_j, \quad \forall j \in P \\
     \sum_{i \in P} x_i &\leq k \\
     \sum_{j \in R} z_j &\geq r,  \\
     \sum_{j \in B} z_j &\geq b,  \\
     z_j, x_i & \in [0,1], \quad \forall i,j \in P.
\end{align*}
        \end{minipage}
    }
    \begin{minipage}{0.1\textwidth}
    \end{minipage}
    \fbox{
        \begin{minipage}[t][6.0cm][t]{0.45\textwidth}
        \begin{center}
        \vspace{2mm}
        \textbf{LP2} 
        \vspace{-4mm}\end{center}
    \begin{align*}
        \text{maximize} \quad  \sum_{j \in S} r_j y_{j} &\\[2mm]
        \text{subject to} \quad \ \  \sum_{j\in S} b_j y_j &\geq b,  \\
        \sum_{j\in S} y_j& \leq k, \\
        y_j  \in &[0,1] \qquad \forall j\in S.
    \end{align*}
        \end{minipage}
    }
    \caption{The linear programs used in the pseudo-approximation algorithm.}
    \label{fig:LPs}
\end{figure}
The algorithm of Bandhyapadhyay et al. \cite{DBLP:conf/esa/Bandyapadhyay0P19} 
first guesses the optimal radius for the instance (there are at most $O(n^2)$ distinct values the optimal radius can take), which we assume by normalization to be one, and considers the natural LP relaxation LP1 depicted on the left in Figure~\ref{fig:LPs}. 
The variable $x_i$ indicates how much point $i$ is fractionally opened as a center and $z_i$ indicates the amount that $i$ is covered by centers. 

Given a fractional solution to LP1, the algorithm of~\cite{DBLP:conf/esa/Bandyapadhyay0P19} finds a clustering of the points. 
The clusters that are produced are of radius two, and with a simple modification (details can be found in Appendix \ref{clusteringproof}), can be made to have a special structure that we call a flower:
\begin{definition}\label{def:flowers}
For $j \in P$, a \textbf{flower} centered at $j$ is the set $\mathcal{F}(j) = \cup_{i\in   \mathcal{B}(j)} \mathcal{B}(i)$. 
\end{definition}
More specifically, given a fractional solution $(x,z)$ to LP1, the clustering algorithm in~\cite{DBLP:conf/esa/Bandyapadhyay0P19}  produces a set of points $S \subseteq P$ and a cluster $C_j \subseteq P$ for every $j\in S$ such that:
\begin{enumerate}
    \item The set $S$ is a subset of  the points $\{j\in P : z_j > 0\}$ with  positive $z$-values.
    \item For each $j\in S$, we have $C_j \subseteq \mathcal{F}(j)$ and the clusters  $\{C_j\}_{j\in S}$ are pairwise disjoint.
    \item If we let $r_j = |C_j \cap R|$ and $b_j = |C_j \cap B|$ for $j\in S$, then the linear program LP2 (depicted on the right in Figure~\ref{fig:LPs}) has a feasible solution $y$ of value at least $r$.
\end{enumerate}
As~LP2 has only two non-trivial constraints, any extreme point will have at most two variables attaining strictly fractional values. So at most $k+1$ variables of $y$ are non-zero. 
The pseudo-approximation of~\cite{DBLP:conf/esa/Bandyapadhyay0P19} now simply takes those non-zero points as centers. 
Since each flower is of radius two,  this gives a $2$-approximation algorithm that opens at most $k+1$ centers. 
(Note that, as the clusters $\{C_j\}_{j\in S}$ are pairwise disjoint, at least $b$ blue points are covered, and at least $r$ red points are covered since the value of the solution is at least $r$.)

Obtaining a constant-factor approximation algorithm that only opens $k$ centers turns out to be significantly more challenging.
Nevertheless, the above techniques form an important subroutine in our algorithm. Given a fractional solution $(x,z)$ to~{LP1}, we proceed as above to find $S$ and an extreme point to~LP2 of value at least $r$. However,  instead of selecting all points with positive $y$-value, we, in the case of two fractional values, only select the one whose cluster covers more blue points. This gives us a solution of at most $k$ centers whose clusters cover at least $b$ blue points. 
Furthermore, the number of red points that are covered is at least $r- \max_{j\in S} r_j$ since we disregarded at most one center. 
As $S \subseteq \{j: z_j >0 \}$ (see first property above) and $C_j \subseteq \mathcal{F}(j)$ (see second property above), we have $\max_{j\in S} r_j \leq \max_{j: z_j > 0} |\mathcal{F}(j) \cap R|$. 
We summarize the obtained properties in the following lemma.
\begin{lemma}
    Given a fractional solution $(x,z)$ to~{LP1}, there is a polynomial-time algorithm that outputs at most $k$ clusters of radius two that cover at least $b$ blue points and at least   $r - \max_{j: z_j > 0} |\mathcal{F}(j) \cap R|$ red points.
    \label{lem:pseudo_approx}
\end{lemma}

We can thus find a $2$-approximate solution that covers sufficiently many blue points but may cover fewer red points than necessary. 
The idea now is that, if the number of red points in any cluster is not too large, i.e., $\max_{j: z_j > 0} |\mathcal{F}(j) \cap R|$ is ``small'',  then we can hope to meet the coverage requirements for the red points by increasing the radius around some opened centers.
Our algorithm builds on this intuition to get a $2$-approximation algorithm using at most $k$ centers for \textit{well-separated} instances as defined below. 

\begin{definition}
An instance of colorful $k$-center is \textbf{well-separated} if there does not exist a ball of radius three that covers at least two balls of $OPT$.
\end{definition}

Our main result of this section can now be stated as follows:
\begin{theorem}
    There is a $2$-approximation algorithm for {well-separated} instances. 
    \label{thm:2approx}
\end{theorem}

The above theorem immediately implies Theorem~\ref{thm:main}, i.e., the $3$-approximation algorithm for general instances. Indeed,
if the instance is not well-separated, we can find a ball of radius three that covers at least two balls of $OPT$ by trying all $n$ points and running the pseudo-approximation of~\cite{DBLP:conf/esa/Bandyapadhyay0P19} on the remaining uncovered points with $k-2$ centers. In the correct iteration, this gives us at most $k-1$ centers of radius two, which when combined with the ball of radius three that covers two balls of $OPT$, is a 3-approximation.

Our algorithm for well-separated instances now proceeds in two phases with the objective of finding a subset of $P$ on which the pseudo-approximation algorithm produces subsets of flowers containing not too many red points. In addition, we maintain a partial solution set of centers (some guessed in the first phase), so that we can expand the radius around these centers to recover the deficit of red points from closing one of the fractional centers. 

\subsection{Phase I}\label{separated}
In this phase we will guess some balls of $OPT$ that can be used to construct a bound on $\max_{j: z_j > 0} |R\cap \mathcal{F}(j)|$. To achieve this, we define the notion of \textbf{Gain}$(p,q)$ for any point $p\in P$ and $q\in \mathcal{B}(p)$.
\begin{definition}
For any $p \in P$ and $q \in \mathcal{B}(p)$, let
\begin{align*}
    \textbf{Gain}(p,q):= R \cap \left( \mathcal{F}(q) \setminus \mathcal{B}(p) \right)
\end{align*}
be the set of \emph{red} points added to $\mathcal{B}(p)$ by forming a flower centered at $q$.
\end{definition}
Our algorithm in this phase proceeds by guessing three centers $c_1, c_2, c_3$ of the optimal solution $OPT$:
\begin{center}
\begin{minipage}{0.95\textwidth}
\begin{mdframed}[hidealllines=true, backgroundcolor=gray!15]
 For $i = 1, 2, 3$, guess the center $c_i$ in $OPT$ and calculate the point $q_i \in \mathcal{B}(c_i)$ such that the number of red points in $\textbf{Gain}(c_i, q_i)\cap P_i$ is maximized over all possible $c_i$, where
\begin{align*}
    P_1 &= P \\
    P_i &= P_{i-1} \setminus \mathcal{F}(q_{i-1}) \mbox{\,\,\,\, for } 2 \leq i \leq 4.
\end{align*}
\end{mdframed}
\end{minipage}
\end{center}
The time it takes to guess $c_1, c_2$, and $c_3$ is $O(n^3)$ and for each $c_i$ we find the $q_i\in \mathcal{B}(c_i)$ such that $|\textbf{Gain}(c_i,q_i) \cap P_i|$ is maximized by  trying all points in $\mathcal{B}(c_i)$ (at most $n$ many).

For notation, define $\textbf{Guess}:= \cup_{i=1}^3 \mathcal{B}(c_i)$ and let 
\begin{align*}
    \tau =  |\textbf{Gain}(c_3,q_3)\cap P_3|.
\end{align*}
The important properties guaranteed by the first phase is summarized in the following lemma.  

\begin{lemma}\label{lemma2}
    Assuming that $c_1, c_2,$ and $c_3$ are guessed correctly, we have  that
    \begin{enumerate}
        \item the  $k-3$ balls  of radius one in $OPT \setminus \{c_i\}_{i=1}^3$ are contained in $P_4$ and cover $b - |B \cap \textbf{Guess}|$ blue points and $r - |R \cap \textbf{Guess}|$ red points; and
        \item  the three clusters $\mathcal{F}(q_1),\mathcal{F}(q_2)$, and $\mathcal{F}(q_3)$ are contained in $P \setminus P_4$ and cover at least $|B \cap \textbf{Guess}|$ blue points and at least $|R \cap \textbf{Guess}| + 3 \cdot \tau$ red points.
    \end{enumerate}
    \label{lem:phaseI}
\end{lemma}
\begin{proof}
1) We claim that the intersection of any ball of $OPT \setminus \{ c_i \}_{i=1}^3$ with $\mathcal{F}(q_i)$ in $P$ is empty, for all $1 \leq i \leq 3$. 
Then the $k-3$ balls in $OPT \setminus \{ c_i \}_{i=1}^3$ satisfy the statement of (1).
To prove the claim, suppose that there is $p \in OPT \setminus \{ c_i \}_{i=1}^3$ such that $\mathcal{B}(p) \cap \mathcal{F}(q_i) \neq \emptyset$ for some $1 \leq i \leq 3$.
Note that $\mathcal{F}(q_i) = \cup_{i\in   \mathcal{B}(q_i)} \mathcal{B}(i)$, so this implies that $\mathcal{B}(p) \cap \mathcal{B}(q') \neq \emptyset$, for some $q'\in \mathcal{B}(q_i)$.
Hence, a ball of radius three around $q'$ covers both $\mathcal{B}(p)$ and $\mathcal{B}(c_i)$ as $c_i \in \mathcal{B}(q_i)$, which contradicts that the instance is well-separated.

2) Note that for $1 \leq i \leq 3$, $\mathcal{B}(c_i) \cup \textbf{Gain}(c_i, q_i) \subseteq \mathcal{F}(q_i)$, and that $\mathcal{B}(c_i)$ and \textbf{Gain}($c_i, q_i$) are disjoint. 
The balls $\mathcal{B}(c_i)$ cover at least $|B \cap \textbf{Guess}|$ blue points and $|R \cap \textbf{Guess}|$ red points, while $\sum_{i=1}^3 |\textbf{Gain}(c_i, q_i) \cap P_i| \geq 3\tau$.
\end{proof}

\subsection{Phase II}
Throughout this section we assume $c_1, c_2$, and $c_3$ have been guessed correctly in Phase I so that the properties of Lemma~\ref{lem:phaseI} hold.
Furthermore, by the selection and the definition of $\tau$, we also have 
\begin{align}
    |\mathbf{Gain}(p, q) \cap P_4| \leq \tau \qquad \mbox{for any $p\in P_4 \cap OPT$ and $q\in \mathcal{B}(p) \cap P_4$. }
    \label{eq:gain_bound}
\end{align}
This implies that  $\mathcal{F}(p) \setminus \mathcal{B}(p)$ contains at most $\tau$ red points of $P_4$.
However, to apply Lemma~\ref{lem:pseudo_approx} we need that the number of red points of $P_4$ in the whole flower $\mathcal{F}(p)$ is bounded.  
To deal with balls with many more than $\tau$ red points, we will iteratively remove \textit{dense} sets from $P_4$ to obtain a subset $P_s$ of \textit{sparse} points. 
\begin{definition}\label{denseset}
When considering a subset of the points $P_s \subseteq P$, we say that a point $j\in P_s$ is \textbf{dense} if the ball $\mathcal{B}(j)$ contains strictly more than $2\cdot\tau$ red points of $P_s$. For a dense point $j$, we also let $I_j \subseteq P_s$ contain those points $i \in P_s$ whose intersection $\mathcal{B}(i) \cap \mathcal{B}(j)$ contains strictly more than $\tau$ red points of $P_s$. 
\end{definition}
We remark that in the above definition, we have in particular that $j \in I_j$ for a dense point $j\in P_s$. Our iterative procedure now works as follows:
\begin{center}
\begin{minipage}{0.99\textwidth}
\begin{mdframed}[hidealllines=true, backgroundcolor=gray!15]
    Initially, let $I = \emptyset$ and $P_s = P_4$. 
    While there is a dense point $j\in P_s$:
    \begin{itemize}
        \item Add $I_j$ to $I$ and update $P_s$ by removing the points $D_j = \cup_{i \in I_j} \mathcal{B}(i) \cap P_s$.
    \end{itemize}
\end{mdframed}
\end{minipage}
\end{center}
Let $P_d = P_4 \setminus P_s$ denote those points that were removed from $P_4$. We will cluster the two sets  $P_s$ and $P_d$ of points separately. Indeed, the following lemma says that a center in $OPT \setminus \{c_i\}_{i=1}^3$ either covers points in $P_s$ or $P_d$ but not points from both sets.  Recall that  $D_j$ denotes the set of points that are removed from $P_s$ in the iteration when $j$ was selected  and so $P_d = \cup_j D_j$.
\begin{lemma}\label{partitionlemma}
For any $c\in$ $OPT\setminus \{c_i\}_{i=1}^3$ and any $I_j\in I$, either $c \in I_j$ or $\mathcal{B}(c)\cap D_j = \emptyset$.
\end{lemma}
\begin{proof}
Let $c\in OPT\setminus \{c_i\}_{i=1}^3$, $I_j\in I$, and suppose $c \notin I_j$.
If $\mathcal{B}(c) \cap D_j \neq \emptyset$, there is a point $p$ in the intersection $\mathcal{B}(c) \cap \mathcal{B}(i)$ for some $i \in I_j$. 
Suppose first that $\mathcal{B}(c) \cap \mathcal{B}(j) \neq \emptyset$. Then, since $c \notin I_j$, the intersection $\mathcal{B}(c) \cap \mathcal{B}(j)$ contains fewer than $\tau$ red points from $D_j$ (recall that $D_j$ contains the points of $\mathcal{B}(j)$ in $P_s$ at the time $j$ was selected). 
But by the definition of dense clients, $\mathcal{B}(j)\cap D_j$ has more than $2 \cdot \tau$  red points, so $(\mathcal{B}(j) \setminus \mathcal{B}(c)) \cap D_j $ has more than $\tau$ red points.
This region is a subset of $\mathbf{Gain}(c,p) \cap P_4$, which contradicts~\eqref{eq:gain_bound}.
This is shown in Figure \ref{fig-lemma2}(a).
Now consider the second case when $\mathcal{B}(c) \cap \mathcal{B}(j) = \emptyset$ and there is a point $p$ in the intersection $\mathcal{B}(c) \cap \mathcal{B}(i)$ for some $i\in I_j$ and $i \neq j$. Then, by the definition of $I_j$, $\mathcal{B}(i) \cap \mathcal{B}(j)$ has more than $\tau$ red points of $D_j$. However, this  is also a subset of $\mathbf{Gain}(c,p) \cap P_4$ so we reach the same contradiction.
See Figure \ref{fig-lemma2}(b).
\end{proof}
\colorlet{mag1}{magenta!10}
\colorlet{mag2}{magenta!40}
\begin{figure}[t]
\begin{center}
\begin{tikzpicture}[scale=0.4]
\draw [fill=mag2] (0, 2) circle [radius=2];
\draw (3, 2) circle [radius=2];

\begin{scope}
  \clip (3, 2) circle [radius=2];
  \fill[white] (0, 2) circle [radius=2];
\end{scope}
\draw [fill] (0, 2) circle [radius=0.05];
\draw [fill] (1.5, 2) circle [radius=0.05];
\draw [fill] (3, 2) circle [radius=0.05];

\node at (3, 1.6) {$c$};
\node at (0, 1.6) {$j$};
\node at (1.5, 1.6) {$p$};

\draw [fill=mag1](13.5, 2) circle [radius=2];
\draw (16, 2) circle [radius=2];
\begin{scope}
  \clip (13.5, 2) circle [radius=2];
  \fill[mag2] (11, 2) circle [radius=2];
\end{scope}
\begin{scope}
  \clip (16, 2) circle [radius=2];
  \fill[white] (13.5, 2) circle [radius=2];
\end{scope}
\draw (11, 2) circle [radius=2];

\draw [fill] (11, 2) circle [radius=0.05];
\draw [fill] (13.5, 2) circle [radius=0.05];
\draw [fill] (16, 2) circle [radius=0.05];
\draw [fill] (14.5, 2) circle [radius=0.05];
\node at (16, 1.6) {$c$};
\node at (11, 1.6) {$j$};
\node at (13.5, 1.6) {$i$};
\node at (14.5, 1.6) {$p$};

\node at (-3,4) {(a)};
\node at (8,4) {(b)};

\end{tikzpicture}
\end{center}
\caption{The shaded regions are subsets of \textbf{Gain}(c,p), which contain the darkly shaded regions that have $> \tau$ red points.} \label{fig-lemma2}
\end{figure}
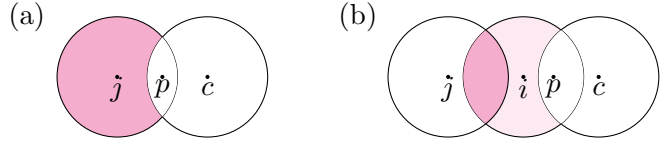

Our algorithm now proceeds by guessing the number $k_d$ of balls of $OPT \setminus \{c_i\}_{i=1}^3$ contained in $P_d$. We also guess the numbers $r_d$ and $b_d$ of red and blue points, respectively, that these balls cover in $P_d$.
Note that after guessing $k_d$, we know that the number of balls in  $OPT\setminus \{c_i\}_{i=1}^3$ contained in $P_s$ equals $k_s = k- 3 - k_d$. Furthermore, by the first property of Lemma~\ref{lem:phaseI}, these balls cover at least $b_s = b - |B \cap \mathbf{Guess}| - b_d$ blue points in $P_s$ and at least $r_s = r - |R \cap \mathbf{Guess}| - r_d$ red points in $P_s$. As there are $O(n^3)$ possible values of $k_d, b_d$, and $r_d$ (each can take a value between $0$ and $n$) we can try all possibilities by  increasing the running time by a multiplicative factor of $O(n^3)$. Henceforth, we therefore assume that we have guessed those parameters correctly. In that case, we show that we can recover an equally good solution for $P_d$ and a solution for $P_s$ that covers $b_s$ blue points and almost $r_s$ red points:
\begin{lemma}\label{algorithms}
    There exist two polynomial-time algorithms $\mathcal{A}_d$ and $\mathcal{A}_s$ such that if $k_d, r_d$, and $b_d$ are guessed correctly  then
    \begin{itemize}
        \item $\mathcal{A}_d$ returns $k_d$ balls of radius one that cover  $b_d$ blue points  of $P_d$ and $r_d$ red points of $P_d$;
        \item $\mathcal{A}_s$ returns $k_s$ balls of radius two that cover at least $b_s$ blue points of $P_s$ and at least $r_s - 3\cdot \tau$ red points of $P_s$.
    \end{itemize}
    \label{lem:twoalgs}
\end{lemma}
\begin{proof}
We first describe and analyze the algorithm $\mathcal{A}_d$ followed by $\mathcal{A}_s$.

\paragraph{The algorithm $\mathcal{A}_d$ for the dense point set $P_d$.} By Lemma~\ref{partitionlemma}, we have that all $k_d$ balls in $OPT \setminus \{c_i\}_{i=1}^3$ that cover points in $P_d$ are centered at points in $\cup_{j} I_j$. Furthermore, we have that each $I_j$ contains at most one center of $OPT$. This is because every $i \in I_j$ is such that $\mathcal{B}(i) \cap \mathcal{B}(j) \neq \emptyset$ and so, by the triangle inequality, $\mathcal{B}(j,3)$ contains all balls $\{\mathcal{B}(i)\}_{i\in I_j}$.
Hence, by the assumption that the instance is well-separated,  the set $I_j$ contains at most one center of $OPT$. 

We now reduce our problem to a $3$-dimensional subset-sum problem.
For each $I_j \in I$,  form a group consisting of an item for each $p\in I_j$. The item corresponding to $p\in I_j$ has the $3$-dimensional value vector $(1, |\mathcal{B}(p) \cap D_j \cap B|, |\mathcal{B}(p) \cap D_j \cap R|)$.  
Our goal is to find $k_d$ items such that at most one item per group is selected and their $3$-dimensional vectors sum up to $(k_d, b_d, r_d)$. 
Such a solution, if it exists, can be found by standard dynamic programming that has a table of size $O(n^4)$. For completeness, we provide the recurrence and precise details of this standard technique in Appendix~\ref{sec:dynamic_programming}.
Furthermore, since the $D_j$'s are disjoint by definition, this gives $k_d$ centers that cover $b_d$ blue points and $r_d$ red points in $P_d$, as required in the statement of the lemma. 

It remains to show that such a solution exists. 
Let $o_1, o_2, \ldots, o_{k_d}$ denote the centers  of the balls in $OPT \setminus \{c_i\}_{i=1}^3$ that cover points in $P_d$. Furthermore, let $I_{j_1}, \ldots, I_{j_{k_d}}$ be the sets in $I$ such that $o_i \in I_{j_i}$ for $i\in \{1,\ldots, k_d\}$. 
Notice that by Lemma~\ref{partitionlemma} we have that $\mathcal{B}(o_i) \cap P_d$ is disjoint from $P_d \setminus D_{j_i}$ and  contained in $D_{j_i}$. It follows that the $3$-dimensional vector corresponding to an $OPT$ center $o_i$ equals $(1, |\mathcal{B}(p) \cap P_d \cap B|, |\mathcal{B}(p) \cap P_d \cap R|)$. Therefore, the sum of these vectors corresponding to $o_1, \ldots, o_{k_d}$ results in the vector $(k_d, b_d, r_d)$, where we used that our guesses of $k_d, b_d$, and $r_d$ were correct.

\paragraph{The algorithm $\mathcal{A}_s$ for the sparse point set $P_s$.} Assuming that the guesses are correct we have that $OPT \setminus \{c_i\}_{i=1}^3$ contains $k_s$ balls that cover $b_s$ blue points of $P_s$ and $r_s$ red points of $P_s$. Hence, LP1 has a feasible solution $(x, z)$ to the instance defined by the point set $P_s$, the number of balls $k_s$, and the constraints $b_s$ and $r_s$ on the number of blue and red points to be covered, respectively. Lemma~\ref{lem:pseudo_approx} then says that we can in polynomial-time find $k_s$ balls of radius two such that at least $b_s$ blue balls of $P_s$ are covered and at least  
\begin{align*}
    r_s - \max_{j: z_j >0}  | \mathcal{F}(j) \cap R|
\end{align*}
red points of $P_s$ are covered. Here, $\mathcal{F}(j)$ refers to the flower restricted to the point set $P_s$.  

To prove the the second part of  Lemma~\ref{lem:twoalgs}, it is thus sufficient to show that LP1 has a feasible solution where $z_j = 0$ for all $j\in P_s$ such that $| \mathcal{F}(j) \cap R| > 3\cdot \tau$. In turn, this follows by showing that, for any such $j\in P_s$ with $|\mathcal{F}(j) \cap R| > 3 \cdot \tau$, no point in $\mathcal{B}(j)$ is in $OPT$ (since then $z_j = 0$ in the integral solution corresponding to $OPT$). Such a feasible solution can be found by adding $x_i=0\hspace{0.1cm} \forall i\in \mathcal{B}(j)$ for all such $j$ to LP1.

To see why this holds, suppose towards a contradiction that there is a $c\in OPT$ such that $c\in \mathcal{B}(j)$. First, since there are no dense points in $P_s$, we have that the number of red points in $\mathcal{B}(c) \cap P_s$ is at most $2 \cdot \tau$. Therefore the number of red points of $P_s$ in $\mathcal{F}(j) \setminus \mathcal{B}(c)$ is strictly more than $\tau$. In other words, we have $\tau < |\mathbf{Gain}(c, j) \cap P_s| \leq |\mathbf{Gain}(c, j) \cap P_4|$ which contradicts~\eqref{eq:gain_bound}.
\end{proof}
Equipped with the above lemma we are now ready to finalize the proof of Theorem~\ref{thm:2approx}.

\begin{proof}[Proof of Theorem~\ref{thm:2approx}]
    Our algorithm guesses the optimal radius and the centers $c_1, c_2, c_3$ in Phase I, and $k_d, r_d, b_d$ in Phase II. There are at most $\binom{n}{2}$ choices of the optimal radius, $n$ choices for each $c_i$, and $n+1$ choices of $k_d,r_d, b_d$ (ranging from $0$ to $n$). We can thus try all these possibilities in polynomial time and, since all other steps in our algorithm run in polynomial time, the total running time will be polynomial. 
    The algorithm tries all these guesses and outputs the best solution found over all choices. For the correct guesses, we output a solution with $3+ k_d + k_s = k$ balls of radius at most two. Furthermore, by the second property of Lemma~\ref{lem:phaseI} and the two properties of Lemma~\ref{lem:twoalgs}, we have that
    \begin{itemize}\itemsep2mm
        \item the number of blue points covered is at least $|B \cap \mathbf{Guess}| + b_d + b_s = b$; and
        \item the number of red points covered is at least $|R \cap \mathbf{Guess}| + 3 \tau + r_d + r_s - 3 \tau = r$.
    \end{itemize}
    We have thus given a polynomial-time algorithm that returns a solution where the balls are of radius at most twice the optimal radius.
\end{proof}
\section{Constant Number of Colors} \label{mult-colors}

Our algorithm extends easily to a constant number $\omega$ of color classes $\C_1, \dots, \C_{\omega}$ with coverage requirements $p_1, \dots, p_{\omega}$.
We use the LPs in Fig. \ref{fig:LPs-omega} for a general number of colors, where $p_{j,i}$ in LP2$(\omega)$ indicates the number of points of color class $i$ in cluster $j \in S$.
$S$ is the set of cluster centers obtained from modified clustering algorithm in Appendix \ref{clusteringproof} to instances with $\omega$ color classes.
LP2$(\omega)$ has only $\omega$ non-trivial constraints, so any extreme point has at most $\omega$ variables attaining strictly fractional values, and a feasible solution attaining objective value at least $p_1$ will have at most $k+\omega-1$ positive values.
By rounding up to 1 the fractional value of the center that contains the most number of points of $\C_{\omega}$, we can cover $p_{\omega}$ points of $\C_{\omega}$.
We would like to be able to close the remaining fractional centers, so we apply an analogous procedure to the case with just two colors.

\begin{figure}[t]
\fbox{
    \begin{minipage}[t][5.0cm][t]{0.45\textwidth}
        \begin{center}
        \vspace{2mm}
        \textbf{LP1$(\omega)$} 
        \vspace{-4mm}\end{center}
\begin{align*}
     \sum_{m \in \mathcal{B}(i)} x_m &\geq z_i, \quad \forall i \in P \\
     \sum_{i \in P} x_i &\leq k \\
     \sum_{i \in C_j} z_i &\geq p_j,  \quad  \forall  1 \leq j \leq \omega \\
     z_i, x_i & \in [0,1], \quad \forall i \in P.
\end{align*}
        \end{minipage}
    }
    \begin{minipage}{0.1\textwidth}
    \end{minipage}
\fbox{
    \begin{minipage}[t][5.0cm][t]{0.45\textwidth}
    \begin{center}
        \vspace{2mm}
        \textbf{LP2$(\omega)$} 
        \vspace{-4mm}\end{center}
    \begin{align*}
        \text{maximize} \quad  \sum_{i \in S} p_{1,i} y_{i} &\\[2mm]
        \text{subject to} \quad \ \  \sum_{i\in S} p_{j,i} y_i &\geq p_j,  \quad \forall 2 \leq j \leq \omega \\
        \sum_{i\in S} y_i& \leq k, \\
        y_i  \in &[0,1] \qquad \forall i\in S.
    \end{align*}
    \end{minipage}
}
\caption{Linear programs for $\omega$ color classes.}
    \label{fig:LPs-omega}
\end{figure}

We can guess $3(\omega-1)$ centers of $OPT$ for each of the $\omega-1$ colors whose coverage requirements are to be satisfied.
Then we bound the number of points of each color that may be found in a cluster, by removing dense sets that contain too many points of any one color and running a dynamic program on the removed sets.
The final step is to run the clustering algorithm of \cite{DBLP:conf/esa/Bandyapadhyay0P19} on the remaining points, and rounding to one the fractional center with the most number of points of $\C_1$, and closing all other fractional centers.

In particular, we get a running time with a factor of $n^{O(\omega^2)}$.
The remainder of this section gives a formal description of the algorithm for $\omega$ color classes.

\subsection{Formal Algorithm for $\omega$ colors}

The following is a natural generalization of Lemma \ref{lem:pseudo_approx} and summarizes the main properties of the clustering algorithm of Appendix \ref{clusteringproof} for instances with $\omega$ color classes.

\begin{custom_lemma}{1$'$}
Given a fractional solution $(x,z)$ to LP1$(\omega)$, there is a polynomial-time algorithm that outputs at most $k$ clusters of radius two that cover at least $p_{\omega}$ points of $\C_{\omega}$, and at least $p_i - (\omega - 1)\max_{j:z_j>0} |\mathcal{F}(j) \cap \C_i|$ for $2 \leq i \leq \omega$.
\end{custom_lemma}

Since we may not meet the coverage requirements for $\omega-1$ color classes, it is necessary to guess some balls of $OPT$ for each of those colors, and for each fractional center.
In total we guess $3(\omega-1)^2$ points of $OPT$ as follows: 
\begin{center}
\begin{minipage}{0.95\textwidth}
\begin{mdframed}[hidealllines=true, backgroundcolor=gray!15]
 For $j= 2, \dots, \omega$, for $i = 1, 2, \dots, 3(\omega-1)$ guess the center $c_{j,i}$ in $OPT$ and calculate the point $q_{j,i} \in \mathcal{B}(c_{j,i})$ such that $| \C_j \cap \textbf{Gain}(c_{j,i}, q_{j,i}) \cap P_{j,i}|$ is maximized over all possible $c_{j,i}$, where
\begin{align*}
    P_{j,1} &= P \\
    P_{j,i} &= P_{j, i-1} \setminus \left( \C_i \cap \mathcal{F}(q_{j, i-1}) \right) \mbox{\,\,\,\, for } 2 \leq i \leq 3(\omega-1) + 1.
\end{align*}
\end{mdframed}
\end{minipage}
\end{center}

This guessing takes $O(n^{3(\omega-1)^2})$ rounds.
It is possible that some $c_{j,i}$ coincide, but this does not affect the correctness of the algorithm.
In fact, this can only improve the solution, in the sense that the coverage requirements will be met with fewer than $k$ centers.
Let $k_c$ denote the number of distinct $c_{j,i}$ obtained in the correct guess.
For notation, define
\begin{align*}
    \textbf{Guess} :&= \cup_{j=2}^{\omega} \cup_{i=1}^{3(\omega-1)} \mathcal{B}(c_{j,i}) \\
    \tau_{j} &= \big\vert \C_j\cap \textbf{Gain}(c_{j,3(\omega-1)},q_{j,3(\omega-1)})\cap P_{j,3(\omega-1)} \big\vert.
\end{align*}
To be consistent with previous notation, let
\begin{align*}
    P_4 := P \setminus \cup_{j=2}^{\omega}\cup_{i=1}^{3(\omega-1)}\mathcal{F}(q_{j,i}).
\end{align*}
The important properties guaranteed by the first phase can be summarized in the following lemma whose proof is the natural extension of Lemma \ref{lemma2}.  

\begin{custom_lemma}{2$'$}
    Assuming that $c_{j,i}$ are guessed correctly, we have that
    \begin{enumerate}
        \item the  $k-3(\omega-1)^2$ balls  of radius one in $OPT \setminus \cup_{j=2}^{\omega} \cup_{i=1}^{3(\omega-1)} \{c_{j,i}\}$ are contained in $P_4$ and cover $p_{\omega} - |\C_{\omega} \cap \textbf{Guess}|$ of points in $\C_{\omega}$ and $p_j - |\C_j \cap \textbf{Guess}|$ points of $\C_j$ for $j=2, \dots, \omega$; and
        \item  the clusters $\mathcal{F}(q_{j,i})$ are contained in $P \setminus P_{3(\omega-1) + 1}$ and cover at least $|\C_{\omega} \cap \textbf{Guess}|$ points of $\C_{\omega}$ and at least $|\C_{j} \cap \textbf{Guess}| + 3(\omega-1) \cdot \tau_{j}$ points of $\C_j$.
    \end{enumerate}
\end{custom_lemma}
Now we need to remove points which contain many points from any one of the color classes to partition the instance into dense and sparse parts which leads to the following generalized definition of dense points. 
\begin{custom_def}{4$'$}
When considering a subset of the points $P_s \subseteq P$, we say that a point $p\in P_s$ is $j$-\textbf{dense} if $|\C_j \cap \mathcal{B}(p) \cap P_s| > 2\tau_j$. For a $j$-dense point $p$, we also let $I_p \subseteq P_s$ contain those points $i \in P_s$ such that $|\C_j \cap \mathcal{B}(i) \cap \mathcal{B}(p) \cap P_s| > \tau_j$ , \textbf{for every $2 \leq j \leq \omega$}.
\end{custom_def}
Now we perform a similar iterative procedure as for two colors:
\begin{center}
\begin{minipage}{0.99\textwidth}
\begin{mdframed}[hidealllines=true, backgroundcolor=gray!15]
    Initially, let $I = \emptyset$ and $P_s = P_{3(\omega-1)}$. 
    While there is a $j$-dense point $p\in P_s$ for any $2\leq j \leq \omega$:
    \begin{itemize}
        \item Add $I_p$ to $I$ and update $P_s$ by removing the points $D_p = \cup_{i \in I_p} \mathcal{B}(i) \cap P_s$.
    \end{itemize}
\end{mdframed}
\end{minipage}
\end{center}

As in the case of two colors, set $P_d = P_{3(\omega-1)} \setminus P_s$. By naturally extending Lemma  \ref{partitionlemma} and its proof, we can ensure that any ball of $OPT \setminus \cup_{j=2}^{\omega}\cup_{i=1}^{3(\omega-1)}\{ c_{j,i} \}$ is completely contained in either $P_d$ or $P_s$.
We guess the number $k_d$ of such balls of $OPT$ contained in $P_d$, and guess the numbers $d_1, \dots, d_{\omega}$ of points of $\C_1, \dots, \C_{\omega}$ covered by these balls in $P_d$.
There are $O(n^{\omega+1})$ possible values of $k_d, d_1, \dots, d_{\omega}$ and all the possibilities can be tried by increasing the running time by a multiplicative factor.
The number of balls of $OPT \setminus \cup_{j=2}^{\omega} \cup_{i=1}^{3(\omega-1)} \{c_{j,i}\}$ contained in $P_s$ is given by $k_s = k - k_c - k_d$ and these balls cover at least $s_j = p_j - |\C_j \cap \textbf{Guess}_{all}| - d_j$ points of $\C_j$ in $P_s$, $1 \leq j \leq \omega$.

Assuming that the parameters are guessed correctly we can show, similar to Lemma \ref{algorithms}, that the following holds.

\begin{custom_lemma}{4$'$}
    There exist two polynomial-time algorithms $\mathcal{A'}_d$ and $\mathcal{A'}_s$ such that if $k_d, d_1, \dots d_{\omega}$ are guessed correctly then
    \begin{itemize}
        \item $\mathcal{A'}_d$ returns $k_d$ balls of radius one that cover $d_1, \dots, d_{\omega}$ points of $\C_1, \dots, \C_{\omega}$ of $P_d$;
        \item $\mathcal{A'}_s$ returns $k_s$ balls of radius two that cover at least $s_{1}$ points of $\C_{1}$ of $P_s$ and at least $s_j - 3(\omega - 1)\cdot \tau_{j}$ points of $\C_j$ of $P_s$, $2\leq j \leq \omega$.
    \end{itemize}
\end{custom_lemma}

The algorithm $\mathcal{A'}_d$ proceeds as did $\mathcal{A}_d$, with the modification that the dynamic program is now $(\omega+1)$-dimensional.
Algorithm $\mathcal{A'}_s$, is also similar to $\mathcal{A}_s$, because LP1 has a feasible solution where $z_p=0$ for all $p \in P_s$ such that $|\mathcal{F}(p) \cap \C_{j}| > 3\tau_{j}$ holds for any $2 \leq j \leq \omega$.
Hence, we output a solution with $k_c + k_d + k_s = k$ balls of radius at most two, and
\begin{itemize}
    \item the number of points of $\C_{1}$ covered is at least $|\C_{1} \cap \textbf{Guess}| + d_{1} + s_{1} = p_{1}$; and
    \item the number of points of $\C_{j}$ covered is at least $|\C_j \cap \textbf{Guess}| + 3(\omega - 1)\tau_{j} + d_{j} + s_{j} - 3(\omega - 1)\tau_{j} = p_j$, for all $j=2, \dots, \omega$. 
\end{itemize}
This is a polynomial-time algorithm for colorful $k$-center with a constant number of color classes.

\section{LP Integrality Gaps}

In this section, we present two natural ways to strengthen LP1 and show that they both fail to close the integrality gap, providing evidence that clustering and knapsack feasibility cannot be decoupled in the colorful $k$-center problem.
On one hand, the Sum-of-Squares hierarchy is ineffective for knapsack problems, while on the other hand, adding knapsack constraints to LP1 is also insufficient due to the clustering aspect of this problem.

\subsection{Sum-of-Squares Integrality Gap}\label{SoS}

The Sum-of-Squares hierarchy (equivalently Lasserre \cite{lasserre2001explicit, lasserre2001global}) is a method of strengthening linear programs that has been used in constraint satisfaction problems, set-cover, and graph coloring, to just name a few examples \cite{DBLP:conf/approx/AroraG11, DBLP:journals/corr/abs-1204-5489, DBLP:conf/stoc/Tulsiani09}.
We use the same notation for the Sum-of-Squares hierarchy, abbreviated as SoS, as in Karlin et al. \cite{DBLP:conf/ipco/KarlinMN11}.
For a set $V$ of variables, $\mathcal{P}(V)$ are the power sets of $V$ and $\mathcal{P}_t(V)$ are the subsets of $V$ of size at most $t$. Their succinct definition of the hierarchy makes use of the \textit{shift operator}: for two vectors $x, y \in \mathbb{R}^{\mathcal{P}(V)}$ the \textbf{shift operator} is the vector $x * y \in \mathbb{R}^{\mathcal{P}(V)}$ such that
\begin{align*}
    (x * y)_I = \sum_{J \subseteq V} x_J y_{I \cup J}.
\end{align*}
Analogously, for a polynomial $g(x) = \sum_{I \subseteq V} a_I \prod_{i \in I} x_i$ we have $(g*y)_I = \sum_{J \subseteq V} a_J y_{I \cup J}$.
In particular, we work with the linear inequalities $g_1, \dots, g_m$ so that the polytope to be lifted is
\begin{align*}
    K = \{x \in [0,1]^n : g_{\ell}(x) &\geq 0 \mbox{ for } \ell = 1, \dots, m \}.
\end{align*}
Let $\mathcal{T}$ be a collection of subsets of $V$ and $y$ a vector in $\mathbb{R}^{\mathcal{T}}$.
The matrix $M_{\mathcal{T}}(y)$ is indexed by elements of $\mathcal{T}$ such that
\begin{align*}
    (M_{\mathcal{T}}(y))_{I, J} = y_{I \cup J}.
\end{align*}
We can now define the $t$-th SoS lifted polytope.
\begin{definition}
For any $1 \leq t \leq n$, the $t$-th SoS lifted polytope $SoS^t(K)$ is the set of vectors $y \in [0,1]^{\mathcal{P}_{2t}(V)}$ such that $y_{\emptyset} = 1$, $M_{\mathcal{P}_t(V)}(y) \succeq 0$, and $M_{\mathcal{P}_{t-1}(V)}(g_{\ell} * y) \succeq 0$ for all $\ell$.

A point $x \in [0,1]^n$ belongs to the $t$-th SoS polytope $SoS^t(K)$ if there exists $y \in SoS^t(K)$ such that $y_{\{i\}} = x_i$ for all $i \in V$.
\end{definition}

We use a reduction from Grigoriev's SoS lower bound for knapsack \cite{DBLP:journals/cc/Grigoriev01} to show that the following instance has a fractional solution with small radius that is valid for a linear number of rounds of SoS.
\begin{theorem}[Grigoriev]
At least $\min \{ 2 \lfloor \min \{k/2, n-k/2 \} \rfloor + 3, n \}$ rounds of SoS are required to recognize that the following polytope contains no integral solution for $k \in \mathbb{Z}$ odd.
\begin{align*}
    \sum_{i=1}^n 2w_i &= k \\
     w_i &\in [0,1] \,\,\,\, \forall i.
\end{align*}
\label{grigoriev}
\end{theorem}
Consider an instance of colorful $k$-center with two colors, $8n$ points, $k = n$, and $r = b = 2n$ where $n$ is odd. 
Points $\{4i-3,4i-2, 4i-1,4i \} \forall i\in [2n]$ belong to cluster $C_i$ of radius one.
For odd $i$, $C_i$ has three red points and one blue point and for even $i$, $C_i$ has one red point and three blue points.
A picture is shown in Figure \ref{lasserre-gap}.
In an optimal integer solution, one center needs to cover at least 2 of these clusters while a fractional solution satisfying LP1 can open a center of $1/2$ around each cluster of radius 1.
Hence, LP1 has an unbounded integrality gap since the clusters can be arbitrarily far apart.
This instance takes an odd number of copies of the integrality gap example given in \cite{DBLP:conf/esa/Bandyapadhyay0P19}.

\begin{figure}
\begin{center}
\begin{tikzpicture}[scale=0.5, rotate=90]


\draw (0,0) circle [radius=1];
\draw (3,0) circle [radius=1];
\draw [fill=black] (1.5,1.5) circle [radius=0.025];
\draw [fill=black] (1.5,1.8) circle [radius=0.025];
\draw [fill=black] (1.5,2.1) circle [radius=0.025];
\draw (0,3.5) circle [radius=1];
\draw (3,3.5) circle [radius=1];
\draw (0,6) circle [radius=1];
\draw (3,6) circle [radius=1];
\draw [fill=blue, blue] (-0.4, 0.4) circle [radius=0.1];
\draw [fill=blue, blue] (-0.4, -0.4) circle [radius=0.1];
\draw [fill=blue, blue] (0.4, -0.4) circle [radius=0.1];
\draw [fill=red, red] (0.3, 0.3) rectangle (0.5, 0.5);
\draw [fill=blue, blue] (-0.4, 3.9) circle [radius=0.1];
\draw [fill=blue, blue] (-0.4, 3.1) circle [radius=0.1];
\draw [fill=blue, blue] (0.4, 3.1) circle [radius=0.1];
\draw [fill=red, red] (0.3, 3.8) rectangle (0.5, 4);
\draw [fill=blue, blue] (-0.4, 6.4) circle [radius=0.1];
\draw [fill=blue, blue] (-0.4, 5.6) circle [radius=0.1];
\draw [fill=blue, blue] (0.4, 5.6) circle [radius=0.1];
\draw [fill=red, red] (0.3, 6.3) rectangle (0.5, 6.5);

\draw [fill=red, red] (2.5, 0.3) rectangle (2.7, 0.5);
\draw [fill=red, red] (2.5, -0.5) rectangle (2.7, -0.3);
\draw [fill=red, red] (3.3, -0.5) rectangle (3.5, -0.3);
\draw [fill=blue, blue] (3.4, 0.4) circle [radius=0.1];
\draw [fill=red, red] (2.5, 3.8) rectangle (2.7, 4);
\draw [fill=red, red] (2.5, 3) rectangle (2.7, 3.2);
\draw [fill=red, red] (3.3, 3) rectangle (3.5, 3.2);
\draw [fill=blue, blue] (3.4, 3.9) circle [radius=0.1];
\draw [fill=red, red] (2.5, 6.3) rectangle (2.7, 6.5);
\draw [fill=red, red] (2.5, 5.5) rectangle (2.7, 5.7);
\draw [fill=red, red] (3.3, 5.5) rectangle (3.5, 5.7);
\draw [fill=blue, blue] (3.4, 6.4) circle [radius=0.1];

\draw [dashed] (-1.25, -1.25) rectangle (4.25, 1.25);
\draw [dashed] (-1.25, 2.25) rectangle (4.25, 4.75);
\draw [dashed] (-1.25, 4.75) -- (-1.25, 7.25);
\draw [dashed] (-1.25, 7.25) -- (4.25, 7.25);
\draw [dashed] (4.25, 4.75) -- (4.25, 7.25);

\draw [decorate, decoration={brace, amplitude=15pt, mirror}, xshift=0pt, yshift=0pt] (5, -1.25) -- (5,7.25 );

\node at (6.5, 3) {$n$};

\end{tikzpicture}
\end{center}
\caption{Integrality gap example for linear rounds of SoS}\label{lasserre-gap}
\vspace{0.5cm}
\end{figure}
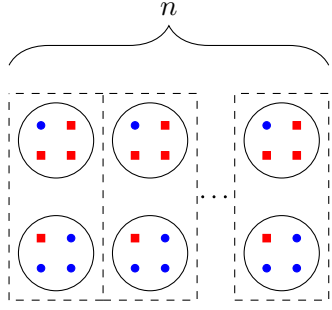

We can do a simple mapping from a feasible solution for the $t$th round of SoS on the system of equations in Theorem \ref{grigoriev} to our variables in the $t$th round of SoS on LP1 for this instance to demonstrate that the infeasibility of balls of radius one is not recognized.
More precisely, we assign a variable $w_i$ to each pair of clusters of radius one as shown in Figure \ref{lasserre-gap}, corresponding to opening each cluster in the pair by $w_i$ amount.
Then a fractional opening of balls of radius one can be mapped to variables that satisfy the polytope in Theorem \ref{grigoriev}. The remainder of this subsection is dedicated to formally describing the reduction from Theorem \ref{grigoriev}.\\
Let $W$ denote the set of variables used in the polytope defined in Theorem \ref{grigoriev}. 
Let $w$ be in the $t$-th round of SoS applied to the system in Theorem \ref{grigoriev} so that $w$ is indexed by subsets of $W$ of size at most $t$. 
Let $V = V_x\cup V_z$, where $V_x = \{x_1, \dots, x_{8n}\}$ and $V_z =\{ z_1, \dots, z_{8n}\}$, be the set of variables used in LP1 for the instance shown in Figure \ref{lasserre-gap}.
We define vector $y$ with entries indexed by subsets of $V$, and show that $y$ is in the $t$-th SoS lifting of LP1. 
In each ball we pick a representative $x_i$, $i \equiv 1 \mod 4$, to indicate how much the ball is opened, so we set $y_I = 0$ if $x_j \in I$, $j \not \equiv 1 \mod 4$.
Otherwise, we set $y_I = w_{\pi(I)}$ where
\begin{align*}
\pi(I) &= \{w_i : x_{8i-3} \mbox{ or } x_{8i-7} \mbox{ or } z_{8i-j} \in I, \mbox{ for some } i \in [n], j \in [7] \}.
\end{align*}
 
We have $M_{\mathcal{P}_t(W)}(w) \succeq 0$, and for $g_1 = -n + \sum_{i=1}^n 2x_i$ and $g_2 = n - \sum_{i=1}^n 2x_i$, $M_{\mathcal{P}_{t-1}(W)}(g_{\ell} * w) \succeq 0$ for $\ell = 1, 2$ since $w$ satisfies the $t$-th round of SoS.
This implies that $M_{\mathcal{P}_{t-1}(W)}(g_{\ell} * w)$ is the zero matrix.

To show that $M_{\mathcal{P}_t(V)}(y)\succeq 0$, we start with $M_{\mathcal{P}_t(W)}(w)$ and construct a sequence of matrices such that the semidefiniteness of one implies the semidefiniteness of the next, until we arrive at a matrix that is $M_{\mathcal{P}_t(V)}(y)$ with rows and columns permuted, i.e. $M_{\mathcal{P}_t(V)}(y)$ multiplied on the left and right by a permutation matrix and its transpose.
Since the eigenvalues of a matrix are invariant under this operation, $M_{\mathcal{P}_t(W)}(w) \succeq 0$ implies that $M_{\mathcal{P}_t(V)}(y)\succeq 0$.

\begin{lemma}
There exists a sequence of square matrices $M_{\mathcal{P}_t(W)}(w) := M_0$, $M_1$, $M_2$, $\dots$, $M_p$, such that the rank of $M_{i}$ is the same as the rank of $M_{i+1}$, $M_i$ is the leading principal submatrix of $M_{i+1}$ of dimension one less, and $M_p$ is $M_{\mathcal{P}_t(V)}(y)$ with rows and columns permuted.
\label{sos-reduction}
\end{lemma}
\begin{proof}
We claim that this sequence of matrices exists with the following description.
Firstly, the matrix $M_{i+1}$ has one extra row and column than $M_i$, and is the same on the leading principal submatrix of size $M_i$.
Then there are two possibilities:
\begin{enumerate}[label=(\alph*)]
    \item The last row and column of $M_{i+1}$ are all zeroes, or
    \item for some $j$, the last row of $M_{i+1}$ is a copy of the $j$th row of $M_i$, the last column is a copy of the $j$th column of $M_i$, and the last entry is $(M_i)_{j,j}$.
\end{enumerate}
Either way, the rank of $M_{i+1}$ would be the same as the rank of $M_i$.

To prove this claim, it suffices to consider a sequence of indices of the matrix $M_{\mathcal{P}_t(V)}(y)$.
The matrix $M_0$ in our sequence will be the submatrix of $M_{\mathcal{P}_t(V)}(y)$ indexed by the first $k$ indices, where $k$ is the dimension of $M_{\mathcal{P}_t(W)}(w)$, i.e. the number of subsets of $W$ of size at most $t$. 
Each subsequent matrix $M_i$ will be the submatrix of $M_{\mathcal{P}_t(V)}(y)$ indexed by the first $k+i$ indices.
Note that the rows/columns of $M_{\mathcal{P}_t(V)}(y)$ can be considered to be indexed by all the subsets of $V$ of size at most $t$.
With this in mind, consider a sequence of subsets of $V$ of size at most $t$ with the following properties:
\begin{enumerate}
    \item All subsets of $\{x_{8i-7}: i \in [n]\}$ of size at most $t$ form a prefix of our sequence.
    \item Each set index after the first has exactly one more element than some set index that came earlier in the sequence.
\end{enumerate}
It is clear that it is possible to arrange all the subsets of $V$ of size at most $t$ in a sequence to satisfy these properties.
It only remains to show that this sequence produces the desired construction for $M_0, M_1, \dots, M_p$.

We have
\begin{align*}
    \left( M_{\mathcal{P}_t(y)} \right)_{I,J} = y_{I \cup J} = w_{\pi(I \cup J)} = w_{\pi(I), \pi(J)}
\end{align*}
so property (1) guarantees that we begin with $M_0$ being $M_{\mathcal{P}_t(W)}(w)$, up to the correct permutation of subsets of $\{x_{8i-7}: i \in [n]\}$.
Now consider some $k'$th index in the sequence, $k' > k$ where $k$ is the dimension of $M_{\mathcal{P}_t(W)}(w)$.
By property (2), it is of the form $J \cup \{x\}$, where $J$ is one of the first $k' - 1$ indices, and $x \in V$.
There are two cases:
\begin{itemize}
    \item If $x$ is some $x_i$ with $i \not\equiv 1 \mod 4$, then $y_{I_{\ell} \cup J} = 0$ for all $\ell \leq k'$.
    \item Otherwise, $\pi(J \cup \{x\}) = \pi(J)$.
\end{itemize}
In the first case, the matrix constructed from the first $k'$ indices will have property (a), and in the second, property (b).
Finally, it is clear that at each step the dimension of the matrices increases by one, and that it is the leading principal submatrix of the following matrix in the sequence, until we end up with $M_{\mathcal{P}_t(V)}(y)$ (up to some permutation of its rows and columns).
\end{proof}

By the rank-nullity theorem, $M_{i+1}$ has one more 0 eigenvalue than $M_i$, so we can apply the following theorem.
\begin{theorem}[Cauchy's Interlace Theorem]
Let $A$ be a symmetric $n \times n$ matrix and $B$ be a principal submatrix of $A$ of dimension $(n-1) \times (n-1)$.
If the eigenvalues of $A$ are $\alpha_1 \geq \cdots \geq \alpha_n$ and the eigenvalues of $B$ are $\beta_1 \geq \cdots \geq \beta_{n-1}$ then $\alpha_1 \geq \beta_1 \geq \alpha_2 \geq \beta_2 \geq \cdots \geq \alpha_{n-1} \geq \beta_{n-1} \geq \alpha_n$.
\label{cauchy}
\end{theorem}
With $M_{i+1} = A$ and $M_i = B$ as in Theorem \ref{cauchy} we have that $\alpha_n = 0$ (since $M_{i+1}$ and $M_i$ have the same eigenvalues but the dimension of the zero eigenspace of $M_{i+1}$ is one greater than that of $M_i$).
Hence, $M_{i+1}$ has no negative eigenvalues if $M_i$ has no negative eigenvalues.
This is sufficient to show that each matrix in the sequence constructed is positive semidefinite, and concludes the proof that $M_{\mathcal{P}_t(V)}(y)\succeq 0$. 

It remains to show that the matrices arising from the shift operator between $y$ and the linear constraints of our polytope are positive semidefinite.
Let $h_i$ denote the linear inequalities in LP1.
In essence, the corresponding moment matrices $M_{\mathcal{P}_{t-1}(V)}(h_i * y)$ are zero matrices since all $h_i$ are tight for the example in Figure \ref{lasserre-gap}.
Formally, we have
\begin{lemma}
Matrices $M_{\mathcal{P}_{t-1}(V)}(h_{\ell} * y)$ are the zero matrix, for each $h_{\ell}$ a linear constraint from LP1.
\end{lemma}
\begin{proof}
Let $h_{1,j}$ be the linear polynomial that corresponds to the first inequality of LP1 for $j \in P$.
First, if $i \not\equiv 1 \mod 4$, then $y_{I \cup \{x_i \}} = 0$ for any $I \subseteq V$.
Otherwise, we have
\begin{align*}
    (M_{\mathcal{P}_{t-1}}(h_{1j} * y))_{I,J} &= \left(\sum_{i \in \mathcal{B}(j, 1)} y_{I\cup J\cup \{x_{i}\}} \right) - y_{I\cup J \cup \{z_j\}} \\ &= w_{\pi(I\cup J)\cup \pi(x_{i})} - w_{\pi(I\cup J) \cup \pi(z_j)} = 0
\end{align*}
since $\pi(\{x_i\}) = \pi(z_j)$ for $i \in \mathcal{B}(j,1)$, $i \equiv 1 \mod 4$.
For the remaining inequalities of LP1: $h_2$, $h_3$, and $h_4$, we have that $M_{\mathcal{P}_{t-1}(V)}(h_{\ell} * y)$ is the zero matrix because of how we defined the projection onto $w$:
\begin{align*}
    (M_{\mathcal{P}_{t-1}}(h_2 * y))_{I,J} &= ny_{I\cup J} - \sum_{x_j\in V_x} y_{I\cup J \cup \{x_j\}}\\
        & = nw_{\pi(I\cup J)} - \sum_{j=1}^n 2w_{\pi(I\cup J \cup \{w_j\})} \\
        & = (M_{\mathcal{P}_{t-1}}(g_2 * w))_{\pi(I),\pi(J)} = 0 \\
    M_{\mathcal{P}_{t-1}}(h_3 * y))_{I,J} &= M_{\mathcal{P}_{t-1}}(h_4 * y))_{I,J} \\
        & = \left( \sum_{j\in R} y_{I\cup J \cup \{z_j\}} \right) - 2ny_{I\cup J}\\
        &= \left( \sum_{i=1}^{n} 4w_{\pi(I\cup J) \cup \{w_i\}} \right) - 2nw_{\pi(I\cup J)} \\
        &= 2(M_{\mathcal{P}_{t-1}}(g_1 * w))_{\pi(I),\pi(J)} = 0. 
\end{align*}
\end{proof}
This concludes the formal proof of the following theorem.
\begin{theorem}
The integrality gap of LP1 with $8n$ points persists up to $\Omega(n)$ rounds of Sum-of-Squares. \hfill $\square$
\end{theorem}

\subsection{Flow Constraints} \label{flow-section}

In this section we add additional constraints based on standard techniques to LP1. These incorporate knapsack constraints for the fractional centers produced in the hope of obtaining a better clustering and show that this fails to reduce the integrality gap.

We define an instance of a knapsack problem with multiple objectives.
Each point $p \in P$ corresponds to an item with three dimensions: a dimension of size one to restrict the number of centers, $|B \cap \mathcal{B}(p)|$, and $|R \cap \mathcal{B}(p)|$.
We set up a flow network with an $(n+1) \times n \times n \times k$ grid of nodes and we name the nodes with the coordinate $(w,x,y,z)$ of its position.
The source $s$ is located at $(0,0,0,0)$ and we add an extra node $t$ for the sink.
Assign an arbitrary order to the points in $P$.
For the item corresponding to $i \in P$, for each $x \in [n]$, $y \in [n]$, $z \in [k]$:
\begin{enumerate}
    \item Add an edge from $(i, x, y, z )$ to $(i+1, x, y, z)$ with flow variable $e_{i,x,y,z}$.
    \item With $b_i := |B \cap \mathcal{B}(i)|$ and $r_i := |R \cap \mathcal{B}(i)|$, if $z < k$ add an edge from $(i, x, y, z)$ to $(i+1, \min \{x+b_i, n\}, \min\{y+b_i, n\}, z+1)$ with flow variable $f_{i,x, y,z}$.
\end{enumerate}
For each $x \in [b, n]$, $y \in [r, n]$:
\begin{enumerate}
\setcounter{enumi}{2}
    \item Add an edge from $(n+1, x, y, k)$ to $t$ with flow variable $g_{x, y}$.
\end{enumerate}
Set the capacities of all edges to one.
In addition to the usual flow constraints, add to LP1 the constraints
\begin{align}
    x_i &= \sum_{x, y \in [n], z \in [k]} f_{i,x,y,z} \quad \mbox{for all } i \in P \\
    1 - x_i &= \sum_{x, y \in [n], z \in [k]} e_{i,x,y,z} \quad \mbox{for all } i \in P.
\end{align}
We refer to the resulting linear program as LP3. 
Notice that an integral solution to LP1 defines a path from $s$ to $t$ through which one unit of flow can be sent; hence LP3 is a valid relaxation.
On the other hand, any path $P$ from $s$ to $t$ defines a set $C_P$ of at most $k$ centers by taking those points $c$ for which $f_{c,x,y, z} \in P$ for some $x, y$, and $z$. Moreover, as  $t$ can only be reached from a coordinate with $x\geq b$ and $y\geq r$ we have that $\sum_{c\in C_P} |\mathcal{B}(c) \cap B| \geq b$ and $\sum_{c\in C_P} |\mathcal{B}(c) \cap R| \geq r$. 
It follows that $C_P$ forms a solution to the problem of radius one \emph{if the balls are disjoint.} In particular, our integrality gap instances for the Sum-of-Squares hierarchy do not fool LP3.   

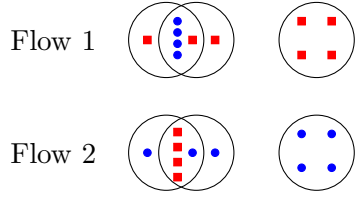
\begin{figure}[t]
\centering
\begin{tikzpicture}[scale=0.5]
\draw (0, 0) circle [radius=1];
\draw (0.8, 0) circle [radius=1];
\draw (4, 0) circle [radius=1];

\draw (0, 3) circle [radius=1];
\draw (0.8, 3) circle [radius=1];
\draw (4, 3) circle [radius=1];

\draw [blue, fill=blue] (-0.5, 0) circle [radius = 0.1];
\draw [blue, fill=blue] (0.7, 0) circle [radius = 0.1];
\draw [blue, fill=blue] (1.3, 0) circle [radius = 0.1];
\draw [red, fill=red] (0.2, 0.45) rectangle (0.4, 0.65);
\draw [red, fill=red] (0.2, 0.05) rectangle (0.4, 0.25);
\draw [red, fill=red] (0.2, -0.35) rectangle (0.4, -0.15);
\draw [red, fill=red] (0.2, -0.75) rectangle (0.4, -0.55);

\draw [red, fill=red] (-0.6, 2.9) rectangle (-0.4, 3.1);
\draw [red, fill=red] (0.6, 2.9) rectangle (0.8, 3.1);
\draw [red, fill=red] (1.2, 2.9) rectangle (1.4, 3.1);
\draw [blue, fill=blue] (0.3, 3.2) circle [radius = 0.1];
\draw [blue, fill=blue] (0.3, 3.5) circle [radius = 0.1];
\draw [blue, fill=blue] (0.3, 2.9) circle [radius = 0.1];
\draw [blue, fill=blue] (0.3, 2.6) circle [radius = 0.1];

\draw [red, fill=red] (3.5, 2.5) rectangle (3.7, 2.7);
\draw [red, fill=red] (3.5, 3.4) rectangle (3.7, 3.6);
\draw [red, fill=red] (4.3, 2.5) rectangle (4.5, 2.7);
\draw [red, fill=red] (4.3, 3.4) rectangle (4.5, 3.6);

\draw [blue, fill=blue] (3.6, -0.4) circle [radius = 0.1];
\draw [blue, fill=blue] (3.6, 0.5) circle [radius = 0.1];
\draw [blue, fill=blue] (4.4, -0.4) circle [radius = 0.1];
\draw [blue, fill=blue] (4.4, 0.5) circle [radius = 0.1];

\node at (-3, 0) {Flow 2};
\node at (-3, 3) {Flow 1};
\end{tikzpicture}
\caption{$k=3$, $r=b=8$}\label{flow-gap}
\end{figure}

The example in Figure \ref{flow-gap} shows that in an instance where balls overlap, the integrality gap remains large.
Here, the fractional assignment of open centers is $1/2$ for each of the six balls and this gives a fractional covering of 8 red and 8 blue points as required.
This assignment also satisfies the flow constraints because the three balls at the top of the diagram define a path disjoint from the three at the bottom.
By double counting the five points in the intersection of two balls we cover 8 red and 8 blue points with each set of three balls.
Hence, we can send flow along each path.
However, this does not give a feasible integral solution with three centers as any set of three clusters does not contain enough points.
In fact, the four clusters can be placed arbitrarily far from each other and in this way we have an unbounded integrality gap since one ball needs to cover two clusters.

%
%
%
\bibliographystyle{splncs04}
\bibliography{refs}

\begin{subappendices}
\renewcommand{\thesection}{\Alph{section}}
\setcounter{section}{0} 
\section{Dynamic Programming for Dense Points}
\label{sec:dynamic_programming}
In this section we describe the dynamic programming algorithm  discussed in Lemma \ref{algorithms}. As stated in the proof of Lemma \ref{algorithms}, given $I = \cup_j I_j$ and correct guesses for $k_d,b_d,r_d$, we need to find $k_d$ balls of radius one centered at points in $I$ covering $b_d$ blue and $r_d$ red points with at most one point from each $I_j\in I$ picked as a center. To do this, we first order the sets in $I$ arbitrarily as $I= \{I_{j_1},\ldots,I_{j_m}\}, m = |I|$. We create a 4-dimensional table $T$ of dimension $(m,b_d,r_d,k_d)$. $T[m',b',r',k']$ stores whether there is a set of $k'$ balls in the first $m'$ sets of $I$ covering $b'$ blue and $r'$ red points. The recurrence relation for T is
\begin{align*}
    T[0,0,0,0] &= \mbox{True} \\
    T[0, b', r', k'] &= \mbox{False,} \quad \mbox{for any } b', r',k' \neq 0 \\
    T[m', b', r', k'] &= \begin{cases}
            \mbox{True} &\mbox{if } T[m'-1, b', r', k'] = \mbox{True} \\
            \mbox{True} &\mbox{if } \exists c \in I_{j_{m'}} \mbox{ s.t. } T[m'-1,b'',r'',k'-1] = \mbox{True, for} \\
                &b'' = b'- |\mathcal{B}(c) \cap B|, r'' = r'
            - |\mathcal{B}(c) \cap R| \\
            \mbox{False } &\mbox{otherwise}
        \end{cases}.
\end{align*}
The table $T$ has size $O((m+1) \cdot (n+1) \cdot (n+1) \cdot(n+1)) = O(n^4)$ since the first parameter has range from $0$ to $m$, and the other parameters can have value $0$ up to at most $n$. Moreover, since $|I_j| \leq n$ for all $I_j \in I$, we can compute the the whole table in time $O(n^5)$ using e.g. the bottom-up approach. We can also remember the choices in a separate table and so we can find a solution in time $O(n^5)$ if it exists.

\section{The Clustering Algorithm}\label{clusteringproof}
In this section we present the clustering algorithm used in \cite{DBLP:conf/esa/Bandyapadhyay0P19} with a simple modification. The algorithm is described in pseudo-code in Algorithm \ref{clusteringalg}. \\

\begin{algorithm}[H]\label{clusteringalg}
\SetAlgoLined
$S \leftarrow \emptyset$, $P' \leftarrow P$ \\
\While{$P' \neq \emptyset$ and $\underset{j\in P^{'}}{\max}$ $z_j >0$}
    {
    $j \in P'$ be a point with maximum $z_j$: let $S \leftarrow S \cup \{j\}$ \\
    $y_j \leftarrow \min \{1, \sum_{i \in \mathcal{B}(j)} x_i  \}$; $\tilde{z}_j \leftarrow y_j$ \\
    $C_j \leftarrow \mathcal{F}(j)  \cap P'$ \\
    For all $j' \neq j \in C_j$, set $y_{j'} \leftarrow 0$, $\tilde{z}_{j'} \leftarrow \tilde{z}_j$ \\
    $P' \leftarrow P' \setminus C_j$
    }
 \caption{Clustering Algorithm}
\end{algorithm}
\vspace{3mm}

Now we state the theorem which states the properties of this clustering algorithm used in Section \ref{pseudo-approx}.
\begin{theorem}
Given a feasible fractional solution $(x,z)$ to LP1, the set of points $S \subseteq P$ and
clusters $C_j \subseteq P$ for every $j\in S$ produced by Algorithm \ref{clusteringalg} satisfy:
\begin{enumerate}
    \item The set $S$ is a subset of  the points $\{j\in P : z_j > 0\}$ with  positive $z$-values.
    \item For each $j\in S$, we have $C_j \subseteq \mathcal{F}(j)$ and the clusters  $\{C_j\}_{j\in S}$ are pairwise disjoint.
\end{enumerate}
Moreover, if we let $R_j = C_j \cap R$ and $B_j = C_j \cap B$ with $r_j = |R_j|$ and $b_j = |B_j|$ for $j\in S$, then $y$ is a feasible solution to LP2 (depicted on the right in Fig.\ref{fig:LPs}) with objective value at least $r$. 
\end{theorem}
\begin{proof}
The proof of the first statement is clear from the condition in the while loop of the algorithm.\\
For the second statement, observe that, by the definition of $C_j$ as stated in the algorithm, $C_j \subseteq \bigcup_{i \in \mathcal{B}(j)} \mathcal{B}(i) = \mathcal{F}(j)$. Since in each iteration, the cluster is removed from $P'$, the clusters are clearly disjoint.
\\
 In order to prove that $y$ is feasible this we first state some useful observations. 
\begin{itemize}
    \item Firstly, for any $i\in P$ there is at most one $j\in S$ such that $d(i,j)\leq 1$. This is true because if there were $j,j'\in S$ such that both $j,j' \in B(j)$ then, assuming w.l.o.g. $j$ was considered before in the while loop, $j'\in C_j$ and thus $j'$ cannot be in $S$ which is a contradiction.
    \item Secondly, note that for any $j_1\in P$ such that $j_1\in C_j$ for some $j$, then $\tilde{z}_{j} = \tilde{z}_{j_1}\geq z_{j_1}$. This is trivially true if $\tilde{z}_{j}=1$, otherwise $\tilde{z}_{j}=\sum_{i\in \mathcal{B}(j)}x_i \geq z_j\geq z_{j_1}$ where the first inequality follows from LP1 constraints and second inequality from the fact that when $C_j$ was removed, $z_j$ had the highest $z$ value. 
\end{itemize}
Now we show that $y$ is feasible for LP2 with objective value at least $r$. Firstly we show that $\sum_{j\in S}r_j y_j \geq r$. To see this,
\begin{align*}
    \sum_{j\in S}r_j y_j  &= \sum_{j\in S}|R_j| y_j \\
    &= \sum_{j\in S}\sum_{j'\in R_j} \tilde{z}_j \hspace{0.5cm} \text{($y_j = \tilde{z}_j$ for any $j\in S$)} \\
    &\geq \sum_{j\in S}\sum_{j'\in R_j} z_{j'} \hspace{0.5cm} \text{(from second observation, $\tilde{z}_j\geq z_{j'}$ for any $j'\in C_j$)}\\
    &= \sum_{j' \in R : z_{j'} > 0} z_{j'} \hspace{0.5cm} \text{(since $C_j$'s are disjoint and contain all $j$ s.t. $z_j>0$)}\\
    &= \sum_{j' \in R} z_{j'} \geq r \hspace{0.4cm} \text{(since $z$ satisfies LP1))}
\end{align*}
Similarly $\sum_{j\in S}b_j y_j \geq b$. Finally we will show that $\sum_{j\in S} y_j \leq k$,
\begin{align*}
    \sum_{j\in S} y_j &\leq \sum_{j\in S} \sum_{j'\in \mathcal{B}(j)} x_{j'} \hspace{0.5cm} \text{(since $y_j\leq \sum_{j'\in \mathcal{B}(j)}x_{j'}$)}\\
    &\leq \sum_{j'\in P} x_{j'} \hspace{0.5cm} \text{(from the first observation)}\\
    &\leq k \hspace{0.5cm} \text{(since $x$ satisfies LP1)}
\end{align*}
This concludes the proof of the claim that $y$ is a feasible solution to LP2 with objective value at least $r$.
\end{proof}

\end{subappendices}

\end{document}